\documentclass[sigconf]{acmart}

\copyrightyear{2026}
\acmYear{2026}
\setcopyright{cc}
\setcctype{by-nc-nd}
\acmConference[ICS '26]{2026 International Conference on Supercomputing}{July 06--09, 2026}{Belfast, United Kingdom}
\acmBooktitle{2026 International Conference on Supercomputing (ICS '26), July 06--09, 2026, Belfast, United Kingdom}
\acmDOI{10.1145/3797905.3807836}
\acmISBN{979-8-4007-2522-7/2026/07}

\acmSubmissionID{53}

\AtBeginDocument{%
  }
\usepackage{booktabs}   
\usepackage{subcaption} 

\usepackage[linesnumbered,ruled,vlined]{algorithm2e}
\SetKwInput{KwInput}{Input}                
\SetKwInput{KwOutput}{Output} 

\usepackage{amsmath}   
\usepackage{caption}
\usepackage{listings}
\usepackage{xcolor}
\usepackage{graphicx}
\usepackage{xspace}
\usepackage{multirow}
\usepackage{centernot}
\usepackage{amsthm}
\usepackage{fontawesome}
\usepackage{url}
\usepackage{float}

\usepackage{enumitem} 
\setlist{leftmargin=*}

\newcommand{\metacomment}[3]%
{{\color{#1}{\bf #2}:\ #3}}

\usepackage[T1]{fontenc}

\newcommand{\etal}{\textit{et. al.}\xspace}

\newtheoremstyle{bfnote}%
  {}{}
  {\itshape}{}
  {\bfseries}{.}
  { }{\thmname{#1}\thmnumber{ #2}\thmnote{ (#3)}}
\theoremstyle{bfnote}
\newtheorem{defi}{Definition}

\begin{document}

\title{Rethinking Collision Detection on GPU Ray Tracing Architecture}

\author{Durga Mandarapu}
\authornote{This work was conducted during the author’s Ph.D. at Purdue University.}
\affiliation{
    \institution{Purdue University, Lawrence Berkeley National Laboratory}
    \city{Berkeley}
    \state{CA}
    \country{USA}}
\email{dmandara@purdue.edu}

\author{Isaac Fuksman}
\affiliation{
    \institution{Purdue University}
    \city{West Lafayette}
    \state{IN}
    \country{USA}}
\email{ifuksman@purdue.edu}

\author{Artem Pelenitsyn}
\affiliation{
    \institution{Purdue University}
    \city{West Lafayette}
    \state{IN}
    \country{USA}}
\email{apelenit@purdue.edu}

\author{Gilbert Bernstein}
\affiliation{
    \institution{University of Washington}
    \city{Seattle}
    \state{WA}
    \country{USA}}
\email{gilbo@cs.washington.edu}

\author{Milind Kulkarni}
\affiliation{
    \institution{Purdue University}
    \city{West Lafayette}
    \state{IN}
    \country{USA}}
\email{milind@purdue.edu}

\renewcommand{\shortauthors}{Mandarapu et al.}
\begin{abstract}

Discrete Collision Detection (DCD) is a fundamental task in several domains including particle-based physics simulations.
Efficient DCD uses indexing structures such as Bounding Volume Hierarchy (BVH), but accelerating irregular BVH traversals demands meticulous efforts to achieve performance.
Modern GPUs feature Ray Tracing (RT) architecture that provides hardware acceleration for BVH traversal and optimized drivers for BVH construction. 
Recent work has attempted to exploit RT architecture to accelerate DCD on spherical particles by reducing DCD to fixed-radius neighbor search.
However, this reduction breaks down for particles with different radii, necessitating the use of large bounding boxes that result in a higher number of duplicate collisions and poor performance.

To address these limitations, we present Mochi, a new reduction that reformulates DCD on RT architecture by exploiting the symmetry of collision relations to support both uniform and non-uniform spherical particles efficiently. 
Mochi introduces per-object proxy spheres that decouple BVH bounding volumes from the collision search radius, enabling significantly tighter bounding boxes without sacrificing correctness. 
Mochi is provably sound and guarantees that all true collisions are detected.
We integrate Mochi into an end-to-end particle simulation pipeline and evaluate it across large-scale particle workloads, showing consistent speedups over state-of-the-art BVH-based and RT-based DCD implementations. 
Mochi generalizes prior RT-based neighbor search formulations while avoiding their fundamental limitations for non-uniform spheres.

\end{abstract}

%
%

\ccsdesc[500]{Computing methodologies~Collision detection}
\ccsdesc[500]{Computing methodologies~Physical simulation}
\ccsdesc[500]{Computing methodologies~Graphics processors}
\ccsdesc[500]{Computing methodologies~Ray tracing}
%
\keywords{Discrete Collision Detection, GPU Ray Tracing architecture, GPU computing, Rigid spheres, Particle systems}


\maketitle

\section{Introduction}

Discrete collision detection is a fundamental computational task in numerous domains, including physical simulations \cite{breuer2015modeling}, computer graphics~\cite{weller2013new}, and industrial manufacturing~\cite{manufacturing}. 
Given a set of objects, a DCD kernel finds all pairs of colliding (or intersecting) objects at a discrete time step and allows the next stages of the application pipeline to resolve collisions.
To prune the quadratic search space, DCD algorithms commonly index objects using spatial acceleration structures, such as BVH trees.
As such, efficient collision detection involves quickly building high-quality BVH and traversing it to identify colliding object pairs.
However, despite decades of research, accelerating DCD using GPU shader cores remains a challenging task requiring meticulous attention to performance, due to the highly irregular control flow and memory access patterns inherent in parallel BVH traversal~\cite{wang2018, oibvh}.

Recent GPUs feature a dedicated ray tracing architecture that augments traditional shader cores with specialized RT cores for accelerating BVH  traversals and geometric intersection tests---key components of DCD.
Although originally introduced to accelerate real-time rendering, RT architecture has increasingly been repurposed for non-rendering workloads involving irregular tree traversals.
Prior work has demonstrated the effectiveness of RT architecture for accelerating a range of proximity and spatial workloads, including nearest neighbor search (NNS)~\cite{evangelou,rtnn,trueknn,arkade}, outlier detection~\cite{dbscan,rtod}, database indexing~\cite{rtscan,rtindex,rtcudb}, spatial queries~\cite{rayjoin,librts}, and Barnes--Hut simulations~\cite{rtbarneshut}.
Building on the reductions of these RT-accelerated spatial workloads, several recent works~\cite{Zhao_arbitrary, vassilev_2021, zhao2023leveraging,roboticspaper} attempt to accelerate DCD by mapping it to RT architecture.

Prior works on accelerating DCD using RT architecture~\cite{Zhao_arbitrary, vassilev_2021, zhao2023leveraging, roboticspaper} implicitly reduce DCD to a neighbor search problem.
In case of particle-based simulations, these approaches index objects as spheres of radius $r$ and model collision queries as point rays, directly reusing the neighbor search reduction.
This formulation works well when all spheres have the same radius (uniform spheres), but it does not generalize to the more common case of spheres with different radii (non-uniform spheres).
To preserve correctness, they approximate non-uniform spheres to large “enough" uniform spheres, which requires constructing overly conservative bounding boxes when building BVHs.
As a result, BVH quality degrades and a large number of unnecessary collision tests are performed on object pairs that do not actually collide, increasing the overall cost of DCD.
More fundamentally, this reduction overlooks the opportunity to exploit RT architecture’s ability to handle irregular, object-centric queries and fails to optimally support non-uniform spheres typical of practical collision detection workloads.

To address the limitations of existing RT-based DCD formulations, we present \emph{Mochi}, a new reduction that reformulates DCD to efficiently leverage ray tracing architecture for all spherical objects, including non-uniform spheres.
Unlike prior reductions that require both objects in a colliding pair to independently detect the collision, Mochi exploits the symmetry of collision relations in DCD and requires detection by only one of the two objects.
This observation enables the construction of per-object proxy spheres that yield significantly tighter bounding boxes without resorting to a large one-size-fits-all radius, thereby reducing unnecessary intersection tests.
Despite relying on smaller bounding volumes, our reduction is provably sound and guarantees that no true collisions are missed.
We further show that the NNS-based uniform-radius approach used in prior work emerges as a special case of our more general formulation.


To demonstrate the applicability of our reduction in real-world use cases, we integrate Mochi into a Discrete Element Method (DEM) particle simulation pipeline.
Custom collision responses can be seamlessly integrated into this pipeline.
The contributions of this paper are as follows:
\begin{enumerate}

\item We identify a fundamental misconception in prior RT-based approaches to collision detection, showing that reducing DCD to fixed-radius neighbor search breaks down in case of non-uniform spheres.

\item We propose a sound reduction for DCD on RT hardware that exploits the symmetry of collision relations and enables tighter, per-object bounding boxes via proxy spheres.

\item We implement this reformulation as Mochi, an RT-based DCD reduction for large-scale particle simulations, and integrate it into a DEM solver\footnote{https://github.com/MDurgaKeerthi/Mochi-DCD-on-RT}.

\item We evaluate Mochi on multiple large-scale benchmarks, demonstrating significant performance improvements over state-of-the-art BVH- and RT-based DCD approaches.


\end{enumerate}

The remainder of the paper is organized as follows.
Section~\ref{sec:background} reviews background on DCD, GPU RT architecture, and prior RT-based DCD formulations. 
Section~\ref{sec:design} presents the Mochi reduction, including the proxy-sphere construction and a formal correctness proof for non-uniform spheres. 
Section~\ref{sec:particlesim} describes the integration of Mochi into a GPU-based DEM simulation pipeline. 
Section~\ref{sec:evaluation} evaluates Mochi across a range of large-scale particle workloads and compares it against state-of-the-art BVH- and RT-based approaches. 
Section~\ref{sec:discussion} discusses limitations and broader applicability, Section~\ref{sec:relatedwork} surveys related work, and Section~\ref{sec:conclusion} concludes.

\section{Background}
\label{sec:background}

In this section, we provide background to DCD, RT architecture and its programming model, and reduction used in the prior DCD work. 

\subsection{Discrete Collision Detection}
Discrete collision detection is widely used in fluid mechanics~\cite{breuer2015modeling}, computer graphics~\cite{oibvh}, and path planning ~\cite{yershova2007improving}, for simulating real-world object behaviors. 
A DCD kernel takes a scene with multiple objects as input and determines all the pairs of colliding objects.
Exhaustively searching all pairs of objects for collisions takes $\mathcal{O}(n^2)$ time, where $n$ is the number of the objects in the scene.
To employ optimizations and acceleration structures, the DCD pipeline is typically divided into two phases : a broad phase and a narrow phase.
In the broad phase of detecting collisions, the goal is to quickly prune the pairs of objects that may not cause a collision by employing methods such as sweep-and-prune (SAP), spatial or object partitioning techniques, or spatial hashing. 
An example for object partitioning techniques is a BVH as it divides the set of objects providing an $\mathcal{O}(n\ logn)$ time complexity.
The potential colliding pairs of objects from the broad phase are then tested using expensive object intersection tests to determine the actual colliding pairs of objects in the narrow phase.
Section~\ref{sec:relatedwork} describes the optimizations the prior works employ in both the phases. 

Objects in collision detection are commonly represented as either spheres or triangle meshes. 
In this work, we focus on objects with spherical geometry, where all spheres have either the same radius (uniform) or varying radii (non-uniform). 
In both cases, the spheres are assumed to have uniform density.
We use the terms \emph{spherical particles}, \emph{spherical objects}, and \emph{spheres} interchangeably. 
Spherical objects are widely used in particle-based simulations and in applications where complex objects are approximated to collections of simple primitives.
For example, Liu \etal~\cite{donut} uses a group of six spheres to represent a torus (donut), while Sui \etal~\cite{roboticspaper} uses sphere-based approximation of robots.
Prior work on DCD using RT architecture supports only spheres with uniform radius, even though non-uniform spheres are more common in practice. 
For example, in smoothed particle hydrodynamics and granular simulations, particles are often modeled as spheres of varying radii to account for differences in mass, density, and scale~\cite{sph,granular}.

\subsection{Ray Tracing Architecture}
RT architecture is offered by multiple GPU vendors such as Nvidia, AMD, Apple, and Intel. 
In this paper, we focus on Nvidia RT architecture, which comprises of shader cores and third-generation RT cores. 
The main objective of RT architecture is to accelerate the RT algorithm for real-time rendering. 
RT algorithm begins by encapsulating objects in the scene with Axis-Aligned Bounding Boxes (AABBs), building a BVH on the AABBs and then launching rays per pixel in the image plane to determine the shading of the pixel based on the ray-object intersections.
To achieve this, the architecture supports BVH construction through optimized drivers that run on shader cores and BVH traversal on RT cores along with ray-triangle or ray-AABB intersection tests.

\subsubsection{Programming Model}
OptiX~\cite{optix} provides a CUDA-style programming interface to RT architecture by allowing the user to call inbuilt OptiX drivers and custom kernels at various stages in the RT programming pipeline.  
Using an inbuilt OptiX call, BVH can be constructed on triangles, spheres, curves, and AABBs. 
OptiX considers spheres hollow, so we used AABBs to contain our proxy spheres.
The set of programs where the RT pipeline allows the user to write custom code are RayGen, Intersection, AnyHit, ClosestHit, and Miss. 
We explain how we use RayGen and Intersection programs in Section~\ref{sec:particlesim}.
AnyHit and ClosestHit programs allow the user to choose whether to continue or terminate the BVH traversal after an intersection.
These shaders are more useful when ranking intersections, so we do not use them.
We do not build Miss shader either in our RT pipeline.


\subsection{Neighbor search reduction in prior DCD work}

Given a set of data points, query points, and a radius parameter $r$, fixed-radius nearest neighbor search finds all data points within distance $r$ of each query point. 
To map this problem onto RT architecture, prior work~\cite{zellman, evangelou} constructs spheres of radius $r$ each centered around a data point and checks if query points fall inside the spheres.
This formulation checks for query points falling within a distance of $r$ from data points instead of checking for data points falling within a distance of $r$ from query point.
Ignoring this inverted formulation, the prior RT-based DCD work~\cite{zhao2023leveraging} and all open-source implementations~\cite{hiprt,RayTracedSPH} directly employ this NNS reduction to DCD.
Given a set of sphere centers and their fixed radius $r$, they use NNS reduction~\cite{zellman, evangelou} to find all centers that are within a distance of $2r$ of each center and interpret the resulting neighbors as colliding pairs.

The challenge is that the NNS reduction cannot support multiple radius parameters simultaneously, which is required for DCD on non-uniform spheres.
As shown in Figure~\ref{fig:uniform_nonuniform} (right), for a sphere of radius $r_1$, the center of a colliding sphere may lie at a distance of up to $r_1 + r_{\max}$, where $r_{\max}$ is the maximum radius in the dataset. But for a sphere of radius $r_2$, the center of a colliding sphere may lie at a distance of up to $r_2 + r_{\max}$.
Consequently, there is no single fixed search radius that can be used with NNS reduction.
One na\"ive alternative is to run the NNS reduction multiple times, each with a different radius parameter. 
However, this would require constructing and traversing multiple BVHs, which is impractical given the high cost of BVH construction on RT architecture. 
Mochi instead detects all collisions in a single iteration by using non-uniform proxy spheres, making our approach a non-trivial generalization of the NNS-to-RT formulation rather than a direct application.

\begin{figure}[h]
  \centering
    \includegraphics[width=0.95\linewidth]{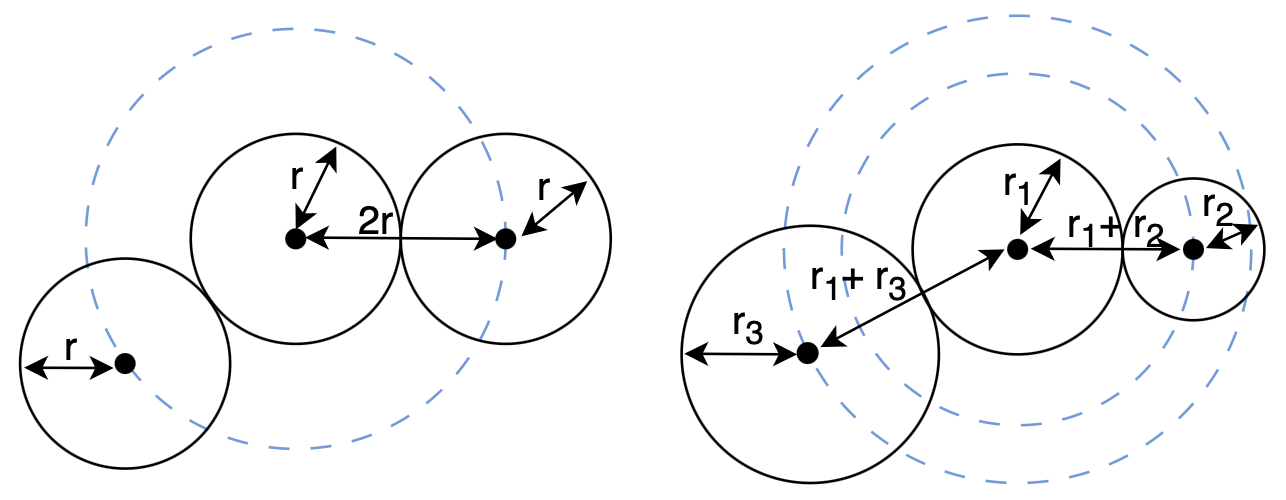}
  \caption{Collisions occur within a distance of $2r$ and $r_i + r_{max}$ from sphere centers for uniform (left) and non-uniform (right) spheres, respectively.}
  \Description{Collisions occur within a distance of 2r and r i plus r max from sphere centers for uniform (left image) and non-uniform (right image) spheres, respectively.}
  \label{fig:uniform_nonuniform}
\end{figure}

Another alternative is to use a single radius which is equal to the largest of all radii and perform the NNS reduction once with it, followed by filtering neighbors to get actual radius neighbors for each query based on the individual particle’s radius.
Prior work adopts this approach by selecting a predetermined search radius large enough to encompass any possible collision pair.
While this avoids false negatives, it results in excessively large bounding volumes and a high number of false positives, necessitating substantial filtering and leading to sub-optimal performance.
By constructing tighter, per-object bounding volumes, Mochi significantly reduces false positives while preserving correctness, improving traversal efficiency without missing collisions.


\section{Collision Detection using Proxy Spheres}
\label{sec:design}


This section describes how Mochi performs collision detection for spherical particles on RT architecture using an intermediate construction of \emph{proxy spheres}.
We focus on DCD for non-uniform spheres, a problem that was not known to be reducible to the RT problem prior to our work.
First, we present the proxy-sphere formulation in the uniform-radius setting to establish intuition (Section~\ref{sec:uniform}), then show why this formulation alone is insufficient, and finally describe how exploiting the symmetry of collision relations enables correct detection for non-uniform spheres (Section~\ref{sec:non-uniform}).
We conclude with a formal correctness proof of the resulting reduction (Section~\ref{sec:proof}).

Rather than treating DCD as a query-centric neighbor search problem, our formulation allows a collision to be detected by either object in a colliding pair. 
Our key observation is that the radii of the spheres in DCD need not have a single correspondence to the search radius used in the neighbor search. 
In particular, decoupling the sphere's bounding volume used in BVH traversal from the actual sphere–sphere collision test enables much tighter per-sphere AABBs than prior approaches.
To capture this fundamental difference, we use proxy spheres to represent a sphere's bounding volume based only on its geometric extent, while relying on ray–proxy-sphere intersections to trigger the collision tests correctly.
Essentially, when two spheres collide, their proxy spheres \emph{together} encompass the collision, ensuring that at least one of the two particles detects it.
We formalize this reduction by first defining proxy spheres in Definition~\ref{def:proxy}.

\begin{defi} [Proxy Spheres] Given a sphere $S$ with center $c$ and radius $r$, the \emph{proxy sphere} corresponding to $S$ (denoted by \textit{proxy-sphere($S$)}) is a sphere with the same center $c$ and radius $2r$.
\label{def:proxy}
\end{defi}

To accelerate DCD on RT architecture, the Mochi reduction starts by constructing a proxy sphere for each input sphere and builds a BVH over the resulting set of proxy spheres. 
Mochi then launches a point ray from the center of every sphere, where a point ray is a ray with infinitesimal length and arbitrary direction. 
These rays traverse the BVH and report ray–proxy-sphere intersections.
Each reported intersection identifies a candidate collision between the sphere corresponding to the ray origin and the sphere corresponding to the intersected proxy sphere. 
Mochi subsequently performs a sphere–sphere intersection test on each such candidate pair to determine whether a true collision has occurred. 
By combining the intersections reported across all rays, Mochi’s DCD kernel identifies all pairs of colliding spheres.

\subsection{Uniform Spheres} 
\label{sec:uniform}

We first consider the special case of uniform-radius spheres, where all spheres have the same radius $r$. 
In this scenario, the Mochi reduction simplifies and coincides with the fixed-radius neighbor search formulation used in prior RT-based DCD work.

For uniform spheres, two spheres $S_i$ and $S_j$ collide if and only if the distance between their centers $\|c_i - c_j\|$ is less than or equal to $2r$. 
This condition holds when a point ray launched from $c_i$ intersects the proxy sphere of $S_j$, since each proxy sphere has radius $2r$. 
Consequently, every ray-proxy-sphere intersection corresponds to a true collision, and every collision is detected by the rays launched from both sphere centers as shown in Figure~\ref{fig:proxyspheres}a.
As a result, in the uniform-radius case, Mochi’s reduction is equivalent to performing a fixed-radius neighbor search with search radius of $2r$, and proxy spheres alone are sufficient to detect all collisions.

\begin{figure}[h]
  \centering
    \includegraphics[width=\linewidth]{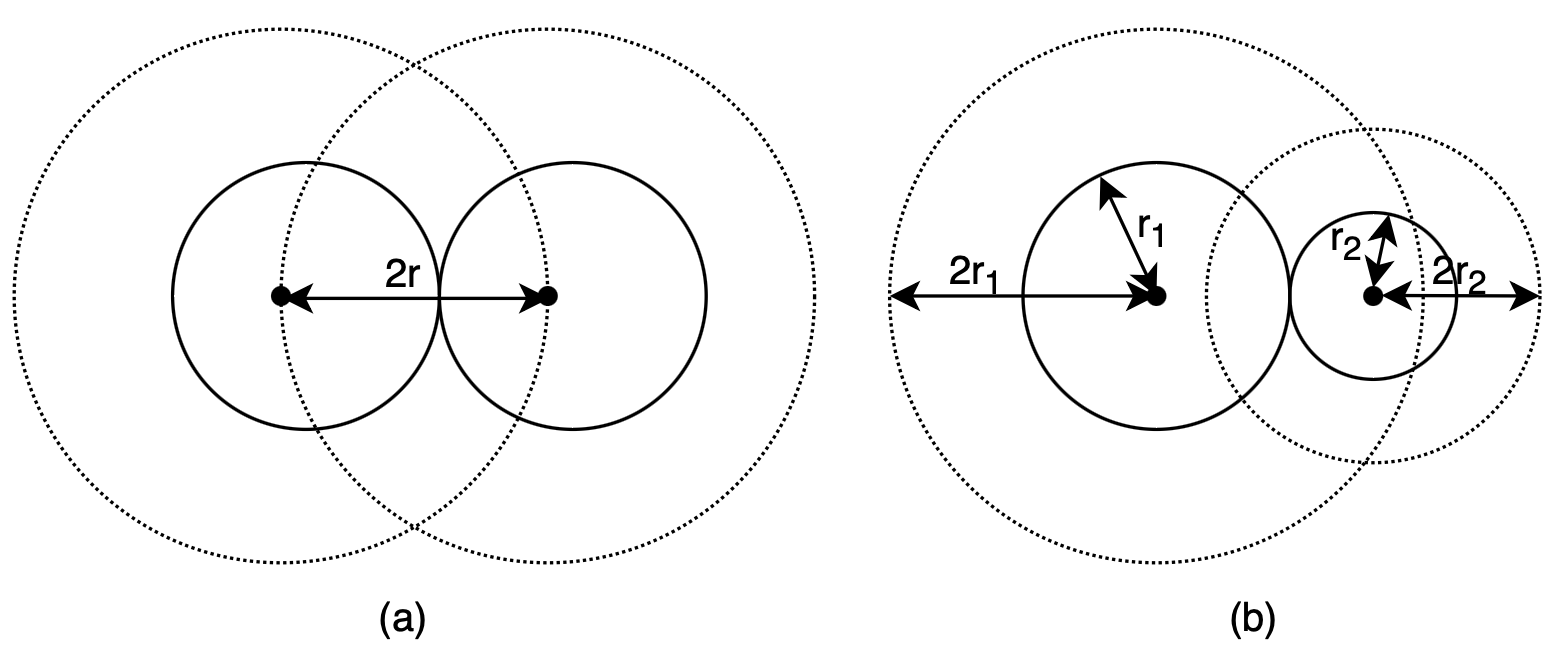}
  \caption{Construction of proxy spheres and point rays. Actual spheres have solid boundary lines while proxy spheres have dotted ones. For uniform spheres, proxy spheres are enough. But in the case of non-uniform spheres, simply using proxy spheres leads to missed collisions (false negatives). }
  \Description{Construction of proxy spheres and point rays. Actual spheres have solid boundary lines while proxy spheres have dotted ones. For uniform spheres, proxy spheres are enough. But in the case of non-uniform spheres, simply using proxy spheres leads to missed collisions (false negatives).}
  \label{fig:proxyspheres}
\end{figure}

\subsection{Non-uniform spheres} 
\label{sec:non-uniform}

However, this equivalence between DCD and neighbor search relies on all spheres sharing a common radius and no longer holds for non-uniform spheres.
For uniform spheres, all collisions involving a sphere of radius $r$ occur within a distance of $2r$ between the centers, enabling a fixed-radius neighbor search formulation.
But when the radii of spheres vary, this characterization breaks down: for a sphere $S_i$ with radius $r_i$, the distance to the center of a colliding sphere also depends on the collision sphere’s radius $r_j$, and can exceed $2r_i$. 
As a result, there is no single search radius associated with $r_i$ or $r_j$ that captures all collisions.

\subsubsection*{Why proxy spheres alone are insufficient for non-uniform spheres}

The key challenge in generalizing the neighbor search reduction to non-uniform spheres is that the radius of the proxy sphere does not correspond to the search radius in NNS.
This mismatch breaks the prior approaches that rely solely on proxy spheres or fixed-radius neighbor search. 
Consider two colliding spheres $S_i$ and $S_j$ with radii $r_i$ and $r_j$.
The distance between their centers is at most $r_i+r_j$, but it has no mathematical relation to $2r_i$ or $2r_j$, the radii of their proxy spheres.
A search radius of $2r_i$ may miss collisions with larger spheres ($r_j > r_i$), while increasing the search radius to accommodate larger spheres introduces excessive false positives. 
Hence, using NNS with $2r_i$ as the search radius or the proxy spheres alone will lead to false positives (if $r_i > r_j$) or false negatives (if $r_i < r_j$). 

As illustrated in Figure~\ref{fig:proxyspheres}b, a point ray launched from the center of a larger sphere may fail to intersect the proxy sphere of a smaller colliding sphere, leading to a missed collision.
Such false negatives are particularly problematic in DCD, as missed collisions compound over time and can result in qualitatively incorrect simulations. 
Figure~\ref{fig:onlyproxy} demonstrates this effect after eight seconds of freefall when using proxy spheres alone.
\footnote{Refer to the first 8 seconds of the supplemental video for complete simulations.}

\begin{figure}[h]
    \centering
    \begin{minipage}{\linewidth}
        \centering
        \begin{subfigure}{0.49\linewidth}
            \includegraphics[height=3.5cm, width=\linewidth]{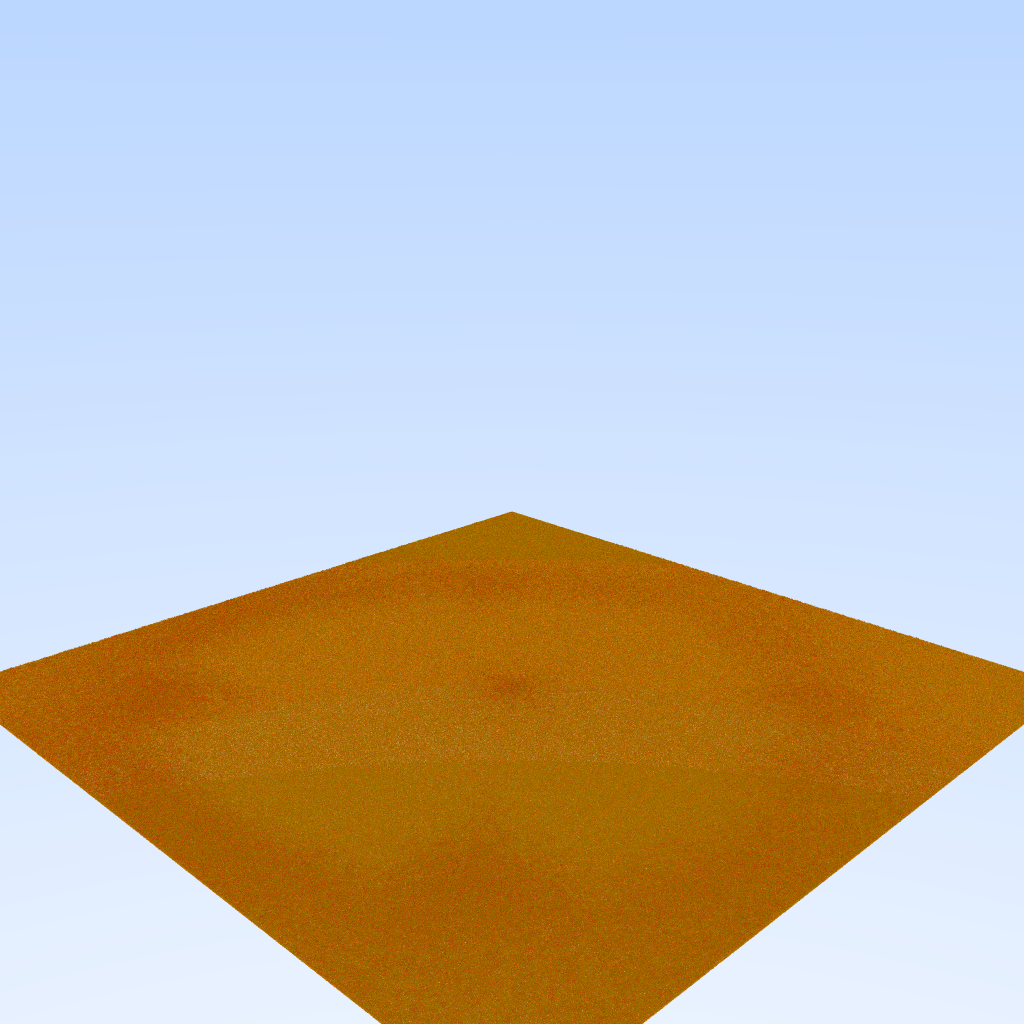}
            \caption{Expected simulation}
            \label{fig:correct_mochi}
        \end{subfigure}
        \hfill
        \begin{subfigure}{0.49\linewidth}
            \includegraphics[height=3.5cm, width=\linewidth]{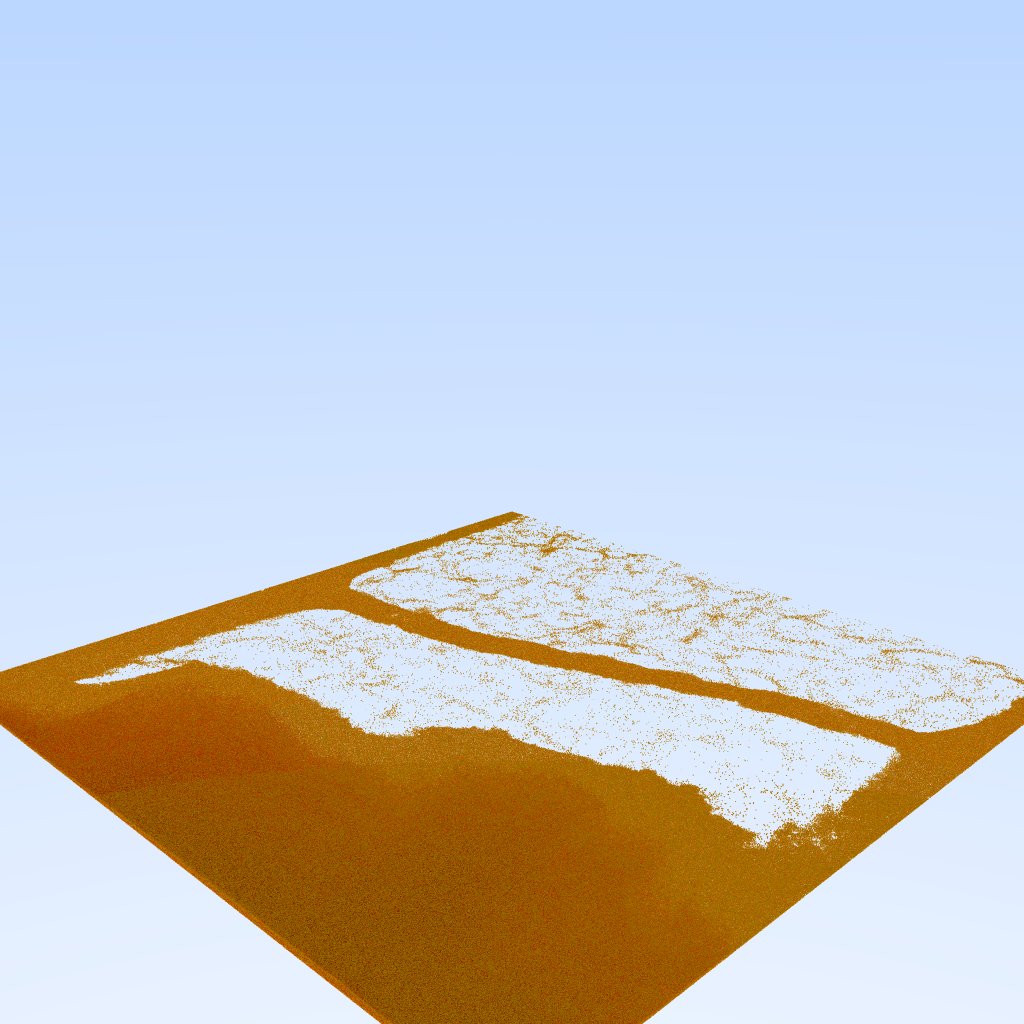}
            \caption{With just proxy spheres}
            \label{fig:incorrect_prior}
        \end{subfigure} 
    \end{minipage}
    \caption{Particles under freefall after eight seconds: using proxy spheres alone results in an incorrect simulation (right) very different from the expected simulation (left).}
    \Description{Particles under freefall after eight seconds: using proxy spheres alone results in an incorrect simulation (right) very different from the expected simulation (left).}
    \label{fig:onlyproxy}
\end{figure}

\paragraph*{Collision Symmetry}

Our key observation is that collision detection is symmetric with respect to the two colliding objects. 
That is, detecting a collision between spheres $S_i$ and $S_j$ does not require both particles to independently detect the collision. It is sufficient if the collision is detected by either of the two spheres.
This symmetry allows Mochi to tolerate per-ray false negatives while still guaranteeing correct per-pair collision detection. 
In particular, while the point ray launched from the center of a larger sphere may fail to intersect the proxy sphere of a smaller colliding sphere, the converse is guaranteed to occur. 
As illustrated in Figure~\ref{fig:proxyspheres}b, the point ray launched from the center of the smaller sphere intersects the proxy sphere of the larger sphere, ensuring that the collision is detected.

\paragraph*{Mochi Reduction for Non-uniform Spheres}
Formally, Mochi applies the same proxy-sphere construction as in the uniform-radius case: for each sphere $S_i$ with center $c_i$ and radius $r_i$, a proxy sphere of radius $2r_i$ is constructed, and a point ray is launched from $c_i$. 
A pair of spheres $(S_i, S_j)$ is reported as colliding if \emph{either} the point ray launched from $c_i$ intersects the proxy sphere of $S_j$ \emph{or} the point ray launched from $c_j$ intersects the proxy sphere of $S_i$, and the subsequent sphere--sphere intersection test confirms a collision.
We leverage the intersections found by both rays corresponding to the collision pair: false negative reported by one particle’s ray will be detected as true positive by the other.
This asymmetric detection rule exploits the symmetry of collisions to ensure that all true collisions are detected, even though individual rays may miss collisions in isolation.


\subsection{Correctness of Mochi reduction}
\label{sec:proof}

We now prove that the Mochi reduction detects all pairs of colliding spheres. 
We begin by formally defining the reduction (Definition~\ref{def:mochi}). 
We then establish a key lemma (Lemma~\ref{lemma:1}) showing that, for any colliding pair of spheres, at least one of the two corresponding point rays intersects the other’s proxy sphere. 
While individual rays may miss collisions, DCD correctness depends on per-pair detection rather than per-ray.
Finally, we use this lemma to prove the correctness of the Mochi reduction (Theorem~\ref{theorem:1}).

\begin{defi} [Mochi reduction] \label{def:mochi}  
Given a set of spheres $\mathcal{S}$, for every sphere $S_i\in \mathcal{S}$ with center $c_i$ and radius $r_i$, \emph{Mochi reduction} builds \emph{proxy-sphere($S_i$)}, launches a point ray from $c_i$ and detects all colliding pairs ($S_i,S_j$) by coordinating the ray-proxy-sphere intersections found by either of the rays launched from $c_i$ and $c_j$.
\label{def:proxy}
\end{defi}

\begin{lemma} [Collision Symmetry] \label{lemma:1}
If two spheres, $S_i$ and $S_j$, with centers $c_i$ and $c_j$ and radii $r_i$ and $r_j$, are colliding, then the point ray launched from $c_i$ will intersect \emph{proxy-sphere($S_j$)}, or vice-versa. 
\end{lemma}

\begin{proof}
We prove by contradiction. We assume that $S_i$ and $S_j$ are colliding and neither of their corresponding rays  intersect the other's proxy sphere.

Since $S_i$ and $S_j$ are colliding, the distance between their centers is at most sum of their radii.
\begin{equation}\label{eqn:1}
    \|c_i-c_j\|\le r_i + r_j
\end{equation}

If the point ray launched from $c_i$ does not intersect \emph{proxy-sphere($S_j$)}, then $c_i$ is outside the proxy sphere and so the distance between $c_i$ and $c_j$, the center of proxy-sphere($S_j$) is greater than the radius of proxy-sphere($S_j$).
\begin{equation}\label{eqn:2}
    \|c_i-c_j\| > 2r_j
\end{equation}

Combining Equations~\ref{eqn:1} and  ~\ref{eqn:2} gives us the following.
\begin{equation}\label{eqn:3}
   r_i + r_j \ge \|c_i-c_j\| > 2r_j \implies r_i > r_j
\end{equation}

Similarly when the point ray launched from $c_j$ does not intersect \emph{proxy-sphere($S_i$)}, we have the following.
\begin{equation}\label{eqn:4}
   r_i + r_j \ge \|c_i-c_j\| > 2r_i \implies r_j > r_i
\end{equation}

Equations~\ref{eqn:3} and ~\ref{eqn:4} contradict.
Hence, our assumption is false, and at least one of the two rays must intersect other’s proxy sphere.
\end{proof}

\begin{theorem}[Correctness of Mochi reduction]
    The Mochi reduction detects \emph{all} pairs of colliding spheres in $\mathcal{S}$.
\end{theorem}\label{theorem:1}

\begin{proof}
    Assume by contradiction that there exists a pair of colliding spheres $(S_i, S_j)$ that is not detected by the Mochi reduction.
    Then by the definition of Mochi reduction (Definition~\ref{def:mochi}), the point ray launched from $c_i$ does not intersect \emph{proxy-sphere($S_j$)}, and the point ray launched from $c_j$ does not intersect \emph{proxy-sphere($S_i$)}. 
    However, this contradicts Lemma~\ref{lemma:1}, which guarantees that at least one of these intersections must occur for any colliding pair. 
    Therefore, all colliding sphere pairs are detected by the Mochi reduction.
\end{proof}

\section{DEM simulation using RT architecture}
\label{sec:particlesim}

This section describes how Mochi implements DCD for particle simulations on GPU RT architecture. 
We integrate Mochi’s DCD kernel into a DEM solver and focus on execution structure rather than physics modeling details.

Mochi executes particle simulation as a sequence of discrete timesteps. 
At each timestep, Mochi (1) builds or refits BVH over proxy spheres, (2) detects collisions by launching rays that traverse the BVH, and (3) resolves collisions and updates state of the particles. 
Algorithm~\ref{alg:mochi} summarizes the resulting simulation loop.

The input to this algorithm is the list of spheres, where each sphere is represented by its center, radius, and additional simulation parameters such as mass, velocity, and collision force.
The output is an array of per-particle collision responses.
The implementation starts in Line~\ref{line:15} by initializing GPU buffers and registering the OptiX programs used in the pipeline: \texttt{OPTIX\_BOUNDS}, \texttt{OPTIX\_INTERSECT}, and \texttt{OPTIX\_RAYGEN}. 
The simulation runs for a total number of timesteps equal to the number of frames multiplied by the number of iterations per frame, with frames rendered every few timesteps to seamlessly support visualization.

\begin{algorithm}
\caption{Mochi's Particle Simulation with DCD}\label{alg:mochi}
\DontPrintSemicolon
  
  \KwInput{Spheres<center, radius, mass, velocity, force>$S$}
  \KwOutput{$\forall S_i \in S, S_i.\text{force} = \sum \text{collision forces on } S_i$}

  \SetKwFunction{Main}{Host\_Main}
  \SetKwFunction{Index}{OPTIX\_BOUNDS}
  \SetKwFunction{DCD}{OPTIX\_INTERSECT}
  \SetKwFunction{Raygen}
  {OPTIX\_RAYGEN}
  
  \SetKwProg{Fn}{Function}{:}{}
  \Fn{\Index{S, threadID}}{
        $thread\_id \gets optixGetLaunchIndex().x$
        $AABB[thread\_id] \gets Box().extend(S[thread\_id].center + 2*radius)$
        $.extend(S[thread\_id].center - 2*radius)$
        \label{line:2} \;
  }

  \SetKwProg{Fn}{Function}{:}{} 
  \Fn{\DCD{S}}{ \tcc{Narrow phase}\label{line:3}
    $ray\_id \gets optixGetLaunchIndex().x$  \label{line:4}\;
    $AABB\_id \gets optixGetPrimitiveIndex()$ \label{line:5}\;
    
    $S_i \gets S[ray\_id]$\; \label{line:6} 
    $S_j \gets S[AABB\_id]$\; \label{line:7}
     
    $dist \gets \|S_i.center-S_j.center\|$ \;\label{line:8}  
    \If{$S_i.radius + S_j.radius - dist \ge 0 $}   { \label{line:9} 
    \tcc{collision is detected}
        \If{$2*S_j.radius - dist \ge 0$ and ($AABB\_id > ray\_id$ or $2*S_i.radius - dist \le 0 $)}     
        { \label{line:10} 
            \tcc{Not a duplicate collision} 
            \text{Resolve\_Collision($S_i, S_j$)}    \label{line:11}  
        }
    }
  }

  \SetKwProg{Fn}{Function}{:}{} 
  \Fn{\Raygen{S}}{\label{line:12}
    $ray\_id \gets optixGetLaunchIndex().x$ \label{line:13}\;
    optixTrace(ray($S[ray\_id]$, vec3f(0,0,1), 1.e-16f), params)
                \label{line:14} \;
  }

  \SetKwProg{Fn}{Function}{:}{}
  \Fn{\Main{}}{\label{line:15}
        optixAccelBuild(S, params) \tcp*{constructs BVH} \label{line:17}
        \For{$i \gets 1$ \KwTo $Num\_Frames$}{ \label{line:18}

            \For{$j \gets 1$ \KwTo $Num\_Iterations\_Per\_Frame$}{
                optixLaunch(S, params) \tcp*{initiates DCD} \label{line:20}
                update(S) \label{line:21}\;
                optixAccelBuild(S, Refit (or Rebuild), params) \label{line:22}\;
            }
        
            Render\_Frame() \label{line:23} \tcp*{Optional}
        } 
  }

\end{algorithm}

\paragraph{BVH maintenance.}
Before the simulation begins, a BVH is constructed over the AABBs of proxy spheres by calling \emph{optixAccelBuild} (Line~\ref{line:17}).
At subsequent timesteps, the BVH is either refit or rebuilt to reflect updated particle positions. 
The same BVH and RT pipeline are reused to render particles for visualization (Line~\ref{line:23}).

\subsection{DCD Implementation}
Given a set of spheres, Mochi’s DCD kernel identifies all colliding pairs of spheres and computes the corresponding collision forces. 
To initiate DCD, we call \emph{optixLaunch} (Line~\ref{line:20}), which invokes the \texttt{OPTIX\_RAYGEN} program (Line~\ref{line:12}). 
This program launches rays in parallel with one point ray per sphere from the sphere center along the $z$-axis. 
Each point ray has an infinitesimal length, and its direction can be arbitrary.

\paragraph{Broad phase.}
The broad phase primarily consists of BVH traversal and ray--AABB intersection tests, both of which are executed on RT cores.
As rays traverse the BVH, each ray--AABB intersection triggers the \texttt{OPTIX\_INTERSECT} program (Line~\ref{line:3}). 
Each such intersection identifies a candidate collision between the sphere from which the ray is launched and the sphere corresponding to the proxy sphere contained in the intersected AABB. 

\paragraph{Narrow phase}
Unlike conventional DCD pipelines that strictly separate broad and narrow phases, Mochi interleaves these phases, thereby avoiding the materialization of large list of candidate collision pairs.
The narrow phase starts when \texttt{OPTIX\_INTERSECT} program is invoked in the broad phase.
For each ray--AABB intersection, Mochi retrieves the sphere corresponding to the ray origin and the sphere corresponding to the intersected AABB (Lines~\ref{line:4}--\ref{line:7}). 
The distance between their centers is then computed (Lines~\ref{line:8}--\ref{line:9}) to determine whether a collision occurred between the spheres.

\paragraph{Duplicate detection.}
To avoid duplicate processing of the same collision pair, Mochi applies two checks in Line~\ref{line:10}. 
First, the ray origin must lie within the proxy sphere of the intersected object, not just the AABB. 
This check is essential because a ray-AABB intersection does not guarantee a ray-proxy-sphere intersection. 
Assume ray $r_i$ detects a collision by hitting AABB of proxy-sphere($S_j$) and ray $r_j$ also detects collision. 
If the first half of \emph{if} condition is absent, both the rays will process the collision because ray $r_j$ only checks if ray $r_i$ hit the proxy-sphere($S_j$) (last part of \emph{if} condition).
Second, when both rays associated with a colliding pair detect the collision, only the ray with smaller ID is allowed to resolve the collision.

\paragraph{Resolving collisions and update.}
When resolving a detected collision (Line~\ref{line:11}), Mochi computes the collision force using standard DEM model equations, also accounting for damping forces. 
After all rays complete detecting and resolving the collisions, Mochi launches a CUDA kernel to apply boundary conditions and update particle properties based on the accumulated forces (Line~\ref{line:21}).
Starting from the narrow phase, all steps are executed on shader cores.





\section{Evaluation}
\label{sec:evaluation}
In this section, we evaluate Mochi to assess when and why it improves DCD on GPU RT architecture. 
Our evaluation is designed to answer four key questions: (1) whether Mochi accelerates DCD kernel independent of simulation overhead when compared to state-of-the-art (SOTA) BVH- and RT-based DCD kernels; (2) how Mochi scales with increase in overall workload and irregularity; (3) where the observed speedups originate, including the relative contributions of BVH construction, traversal, and collision processing; and (4) how sensitive Mochi is to relative particle motion across timesteps and BVH update strategies such as refit versus rebuild. 
These experiments evaluate the impact of Mochi’s DCD reformulation, rather than the effects of tuning BVH maintenance parameters or other implementation specific optimizations.
To this end, we first benchmark Mochi’s DCD kernel in isolation, then evaluate it within full DEM simulations across a range of workloads, and finally analyze its performance characteristics and limitations.

\paragraph{Baselines}
We use three GPU-based implementations as baselines, two of which primarily use shader cores and one that leverages RT architecture.
Zhao et al.~\cite{zhao2023leveraging} is a SOTA RT-based DCD implementation that uses fixed-radius neighbor search with the search radius equal to four times the maximum particle radius.
Oibvh~\cite{oibvh} is a SOTA DCD implementation based on shader cores that uses a BVH as its indexing structure.
Taichi~\cite{taichi} is a domain-specific language for high-performance graphics applications such as DCD; although it exposes a high-level programming interface, its underlying CUDA kernels are heavily optimized.
We use uniform grid and hashmap-based particle simulation implementations provided on Taichi's blog as simulation baselines.
We compare Mochi’s DCD kernel against Zhao \etal~\cite{zhao2023leveraging} and Oibvh~\cite{oibvh}.
For end-to-end simulations, we use Taichi as the baseline.

\paragraph{Experimental setup}
We implemented Mochi using OptiX $8$ and CUDA $12.2$ on NVIDIA GeForce RTX 4070Ti GPU. 
We evaluate Mochi and compare it with a range of baselines using four benchmarks to explore its performance in detail.
We obtained Oibvh source code from authors, varying only the maximum number of the threads the code can issue. 
We could not obtain code from Zhao \etal, so we implemented their DCD approach in our simulation pipeline.

Taichi uses a Python interface to call its optimized CUDA kernels underneath. 
To allow fair comparison, only the time taken by CUDA kernels is counted towards runtime so that Taichi is not penalized due to Python overheads.
Both Taichi and Mochi were warmed up for 200 iterations before timing the simulation, since both use JIT compiler. 
Additionally, Taichi's offline caching was enabled during the experiments, representing best-practice usage.

\paragraph{DEM simulations}
In all simulations, the run time includes build, DCD, and, update times.
The build time refers to the initial indexing of particles and refitting or rebuilding of the index structures after every iteration.
The DCD time accounts for the collision detection time, including the computation time for normal forces once a collision is detected.
The update time includes testing boundary conditions and updating the particle's center and velocity due to collisions.
In simulations, one iteration signifies a timestep, and it's value is described in the benchmarks sections.
Only gravity and normal collision forces are considered, while friction and tangential forces are not computed.

\paragraph{Benchmarks} We considered three datasets: 1) $5$ million (M) particles in a unit cube with their radii uniformly distributed in the range $0.0005$-$0.0006$ and a density of $500$, 2) $1$M particles in a 16-unit cube with three different radii ranges $0.0005$-$0.0006$, $0.0005$-$0.006$, and, $0.0005$-$0.06$, and 3) $14$ bunnies totaling of $1.1$M particles with radius $r(=0.005)$ or $r/2$.
The time steps are $1e$-$4$, $5e$-$5$, and $2.5e$-$5$ seconds, respectively. 
The datasets are used to simulate two motions: freefall under gravity and rotating gravity in $XZ-$plane at an angular velocity of $0.25\pi$ radians per second. 

\subsection{Benchmarking DCD kernel}
In this section, we benchmark Mochi's DCD kernel against the state-of-the-art BVH-based implementations using shader and RT cores, isolating collision detection performance from simulation overhead in order to directly compare different DCD formulations on the GPU.

\subsubsection{Comparing with prior work on RT cores}
In Figure~\ref{fig:prior_power}, we compare the total simulation time of Mochi with Zhao \etal~\cite{zhao2023leveraging} by varying the number of particles in the second dataset with radii in the range $0.0005$-$0.06$ under freefall.
Mochi achieves an overall speedup of $1.04$x-$1.56$x as the number of particles increases from $10$K to $1$M.
We observe that Mochi’s speedups increase with an increase in the number of particles, indicating better scalability of its DCD formulation.

\begin{figure}[h]
    \centering
    \includegraphics[width=0.5\linewidth]{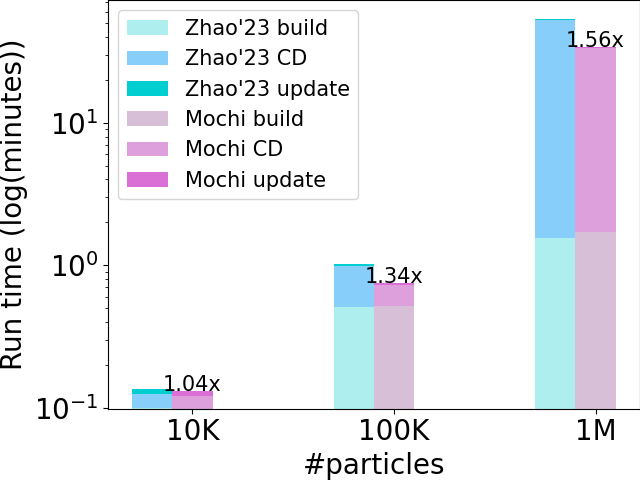}
    \caption{Performance comparison of Mochi DCD kernel against SOTA shader and RT core-based implementations. Freefall on particles in radii range $0.0005$-$0.06$. Then numbers on Mochi bars show Mochi speedups over baseline.}
    \Description{Performance comparison of Mochi DCD kernel against SOTA shader and RT core-based implementations. Freefall on particles in radii range $0.0005$-$0.06$. Then numbers on Mochi bars show Mochi speedups over baseline. The subsubsection referencing this image explains the trends.}
    \label{fig:prior_power}
\end{figure}

In Figure~\ref{fig:prior_radii}, we compare the DCD kernel times under rotating gravity for three different radii ranges in the second dataset.
As the ratio of maximum to minimum particle radius increases, Mochi’s speedup over Zhao \etal increases from $1.0$x to $4.33$x.
This trend is expected: as the fixed-radius RT baseline constructs AABBs of side length four times the maximum particle radius, whereas Mochi constructs per-particle AABBs proportional to each sphere’s radius.
As the ratio between the maximum and minimum radii increases, the inefficiency of using maximum radius as a parameter for AABB size becomes more pronounced.
By using tighter, per-particle AABBs, Mochi significantly reduces unnecessary intersection tests, resulting in increasing speedups with greater radius non-uniformity and higher number of particles.


\begin{figure}[h]
    \centering
    \includegraphics[width=0.5\linewidth]{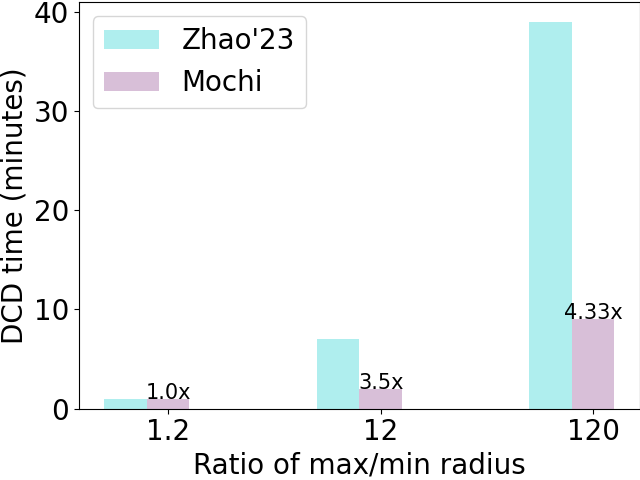}
    \caption{Performance comparison of Mochi DCD kernel against SOTA RT core-based implementation for varying particles radii range. Rotating gravity in $xz$-plane with constant number of $1$M particles. Then numbers on Mochi bars show Mochi speedups over baseline.}
    \Description{Performance comparison of Mochi DCD kernel against SOTA RT core-based implementation for varying particles radii range. Rotating gravity in $xz$-plane with constant number of $1$M particles. Then numbers on Mochi bars show Mochi speedups over baseline. The subsubsection referencing this image explains the trends.}
    \label{fig:prior_radii}
\end{figure}

\subsubsection{Comparing with SOTA BVH}
In Figure~\ref{fig:oibvh}, we compare Mochi with Oibvh~\cite{oibvh}, a state-of-the-art BVH implementation for DCD, using freefall and rotating gravity benchmarks on $1$M particles.
We modified Oibvh code so that collisions between spherical particles are detected instead of triangles, which it was originally built for.
Since Oibvh does not perform simulation, we export particle positions to text files per frame and input these files to Oibvh and Mochi to measure only the build and DCD times.

\begin{figure}[h]
    \centering
    \includegraphics[width=0.5\linewidth]{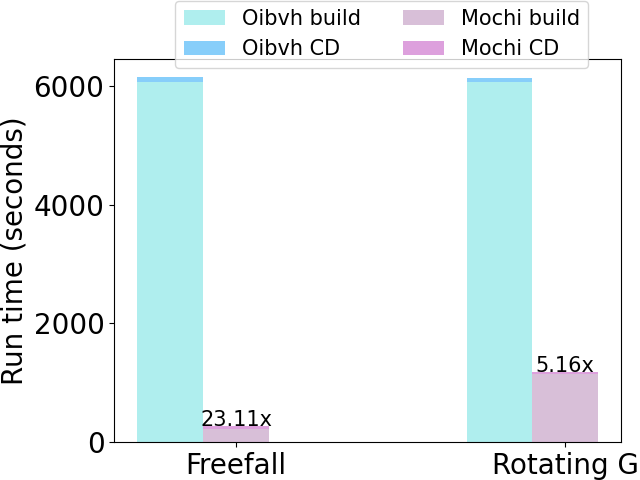}
    \caption{Performance comparison of Mochi DCD kernel against Oibvh across two benchmarks for $1$M particles. Rotating G stands for rotating gravity. Then numbers on Mochi bars show Mochi speedups over baseline.}
    \Description{Performance comparison of Mochi DCD kernel against Oibvh across two benchmarks for $1$M particles. Rotating G stands for rotating gravity. Then numbers on Mochi bars show Mochi speedups over baseline. The subsubsection referencing this image explains the trends.}
    \label{fig:oibvh}
\end{figure}

The overall speedups of Mochi over Oibvh are $23.11$x and $5.16$x, of which build time speedups are $26.3$x and $5.25$x for freefall and rotating gravity benchmarks, respectively.
The higher speedup for freefall benchmark compared to rotating gravity stems from the fact that Mochi refits the BVH, whereas Oibvh rebuilds it every frame.
In DCD times, Mochi outperforms Oibvh by a factor of $2$x for both benchmarks, indicating that the BVH constructed by RT architecture provides better traversal efficiency, despite Oibvh spending several times longer on BVH construction.

\subsection{Benchmarking DEM simulations}

In this section, we benchmark Mochi against Taichi in end-to-end DEM simulations to study the impact of collision detection, BVH maintenance, and particle motion on each other in realistic workloads.

\subsubsection{Freely falling particles}
\label{subsec:freefall}
Figure~\ref{fig:freefall_dataset} shows a frame from the freefall simulation of particles in the first dataset\footnote{Refer to the 8:16 mark of the video for a complete simulation.}.
Figure~\ref{fig:freefall_plot} shows the total runtime of simulations produced by Taichi using uniform grid and hashmap implementations, and by Mochi, as the number of particles varies from $1$M to $5$M.
Each simulation runs for $80{,}000$ iterations.
In every iteration, Mochi refits the BVH, while Taichi rebuilds its grid and hashmap.
For $1$M particles, the mean and standard deviation of runtime over $5$ runs are $31.44$ and $0.026$ for Taichi’s hashmap, and $6.73$ and $0.029$ for Mochi.
For $5$M particles, Taichi’s uniform grid implementation times out.
Across all configurations, Mochi consistently outperforms Taichi’s fastest variant.
We analyze the source of these speedups in Section~\ref{subsec:speedups}.

\begin{figure}[h]
    \centering
    \begin{minipage}{\linewidth}
        \centering
        \begin{subfigure}{0.49\linewidth}
            \includegraphics[width=\linewidth]{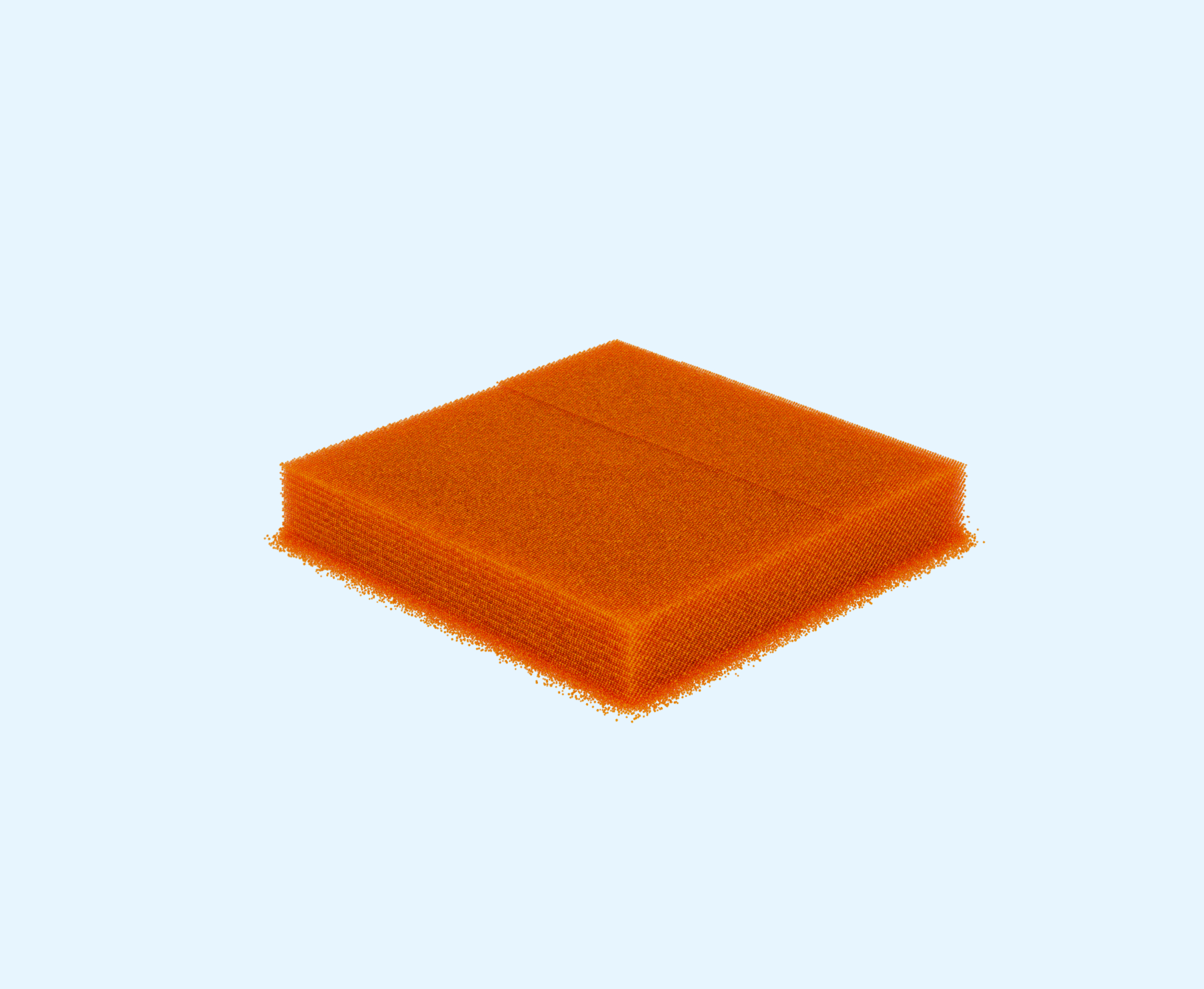}
            \caption{Frame 6 in simulation}
            \label{fig:freefall_dataset}
        \end{subfigure}
        \hfill
        \begin{subfigure}{0.49\linewidth}
            \includegraphics[width=\linewidth]{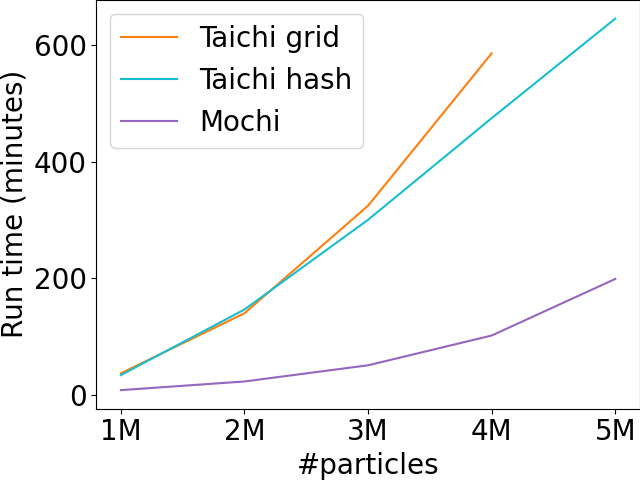}
            \caption{Runtime comparison}
            \label{fig:freefall_plot}
        \end{subfigure} 
    \end{minipage}
    \caption{Performance comparison of Mochi against Taichi for freefall DEM particle simulation of 8 seconds. }
    \Description{Performance comparison of Mochi against Taichi for freefall DEM particle simulation of 8 seconds. The subsubsection referencing this image explains the trends.}
    \label{fig:1afreefall}
\end{figure}

\subsubsection{Rotating gravity}
\label{subsec:rotate}
To induce increased particle motion and a sloshing effect, we rotate gravity (Figure~\ref{fig:rotate_dataset}) and test the indexing structure’s ability to adapt to more dynamic workloads\footnote{Refer to the 16:24 mark of the video for a complete simulation.}.
Unlike the freefall benchmark in Section~\ref{subsec:freefall}, we observe that rebuilding the BVH for Mochi instead of refitting it at every iteration would result in significantly faster overall simulation time.
This indicates that the BVH refit operation offered by RT architecture is best suited for scenarios with limited relative particle displacement, similar to conventional BVH refit strategies.
We could not report how fast the rebuild is, as the refit-based version times out.

\begin{figure}[h]
    \centering
    \begin{minipage}{\linewidth}
        \centering
        \begin{subfigure}{0.49\linewidth}
            \includegraphics[height=3.3cm, width=\linewidth]{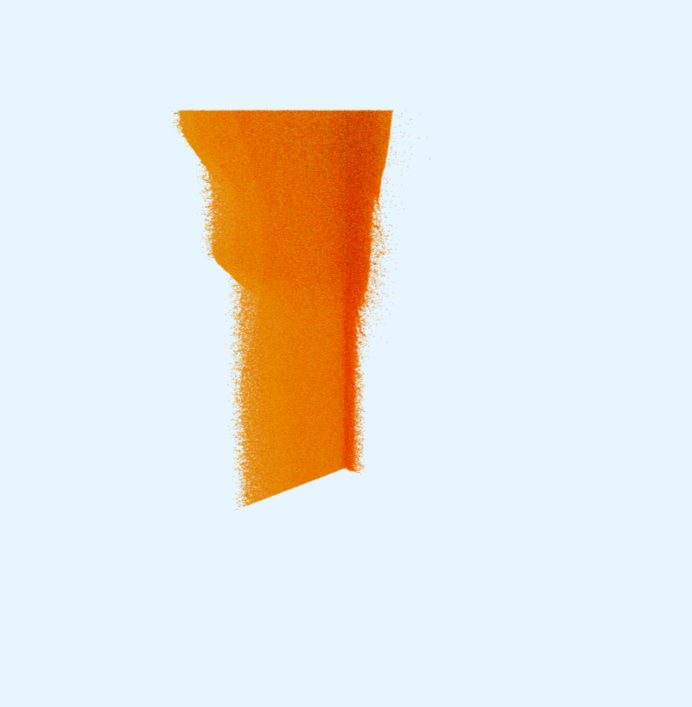}
            \caption{Frame 70 in simulation}
            \label{fig:rotate_dataset}
        \end{subfigure}
        \hfill
        \begin{subfigure}{0.49\linewidth}
            \includegraphics[width=\linewidth]{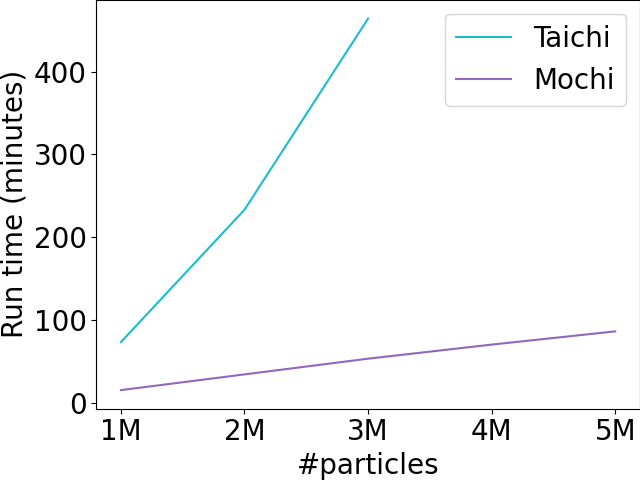}
            \caption{Runtime comparison}
            \label{fig:rotate_plot}
        \end{subfigure} 
    \end{minipage}
    \caption{Performance comparison of Mochi against Taichi for rotating gravity DEM particle simulation of 8 seconds.}
    \Description{Performance comparison of Mochi against Taichi for rotating gravity DEM particle simulation of 8 seconds. The subsubsection referencing this image explains the trends.}
    \label{fig:1b_rotate}
\end{figure}

Figure~\ref{fig:rotate_plot} compares the runtime of Taichi and Mochi when simulating particles from the first dataset under rotating gravity for $160{,}000$ iterations.
In this experiment, the number of collision pairs increases at a higher scale than in the freefall benchmark due to increased particle motion.
For Taichi, we report only for the hashmap implementation in the rest of the evaluation, as the uniform grid version is either slow or times out.
From $1$M to $3$M particles, Mochi achieves speedups of $4.87$x to $8.75$x over Taichi.
Beyond $3$M particles, Taichi’s hashmap runs out of memory, which is expected since it stores all potential collision pairs prior to resolving them.
In contrast, Mochi resolves collisions immediately without storing potential pairs.
Compared to the previous freefall experiment with less particle motion, Mochi achieves higher speedups, highlighting that RT architecture’s ability to efficiently refit BVHs allows Mochi to better handle highly dynamic particle motion.

\subsubsection{Irregular workloads - changing radii}
Figure~\ref{fig:radii_plot} compares the runtimes of Taichi and Mochi when simulating particles with different radii ranges ($0.0005$--$0.006$ and $0.0005$--$0.06$) from the second dataset under rotating gravity for $160{,}000$ iterations (Figure~\ref{fig:radii_dataset})\footnote{Refer to the 24:32 mark of the video for a complete simulation.}.
The time spent in individual phases is shown using a gradient of colors, and the labels on Mochi’s bars indicate speedups over Taichi.
For both Taichi and Mochi, the index times are almost the same across the two radii ranges, indicating that the indexing cost of Taichi’s hashmap and Mochi’s BVH depends primarily on the number of particles rather than their radii.

\begin{figure}[h]
    \centering
    \begin{minipage}{\linewidth}
        \centering
        \begin{subfigure}{0.49\linewidth}
            \includegraphics[height=3.3cm,width=\linewidth]{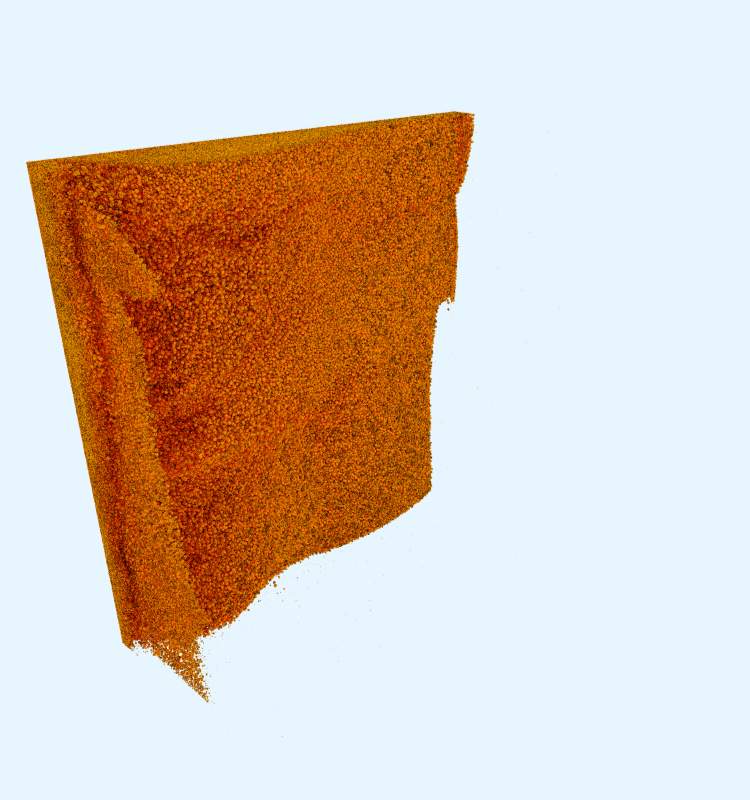}
            \caption{Frame 168 in $0.0005$-$0.06$ setting}
            \label{fig:radii_dataset}
        \end{subfigure}
        \hfill
        \begin{subfigure}{0.49\linewidth}
            \includegraphics[width=\linewidth]{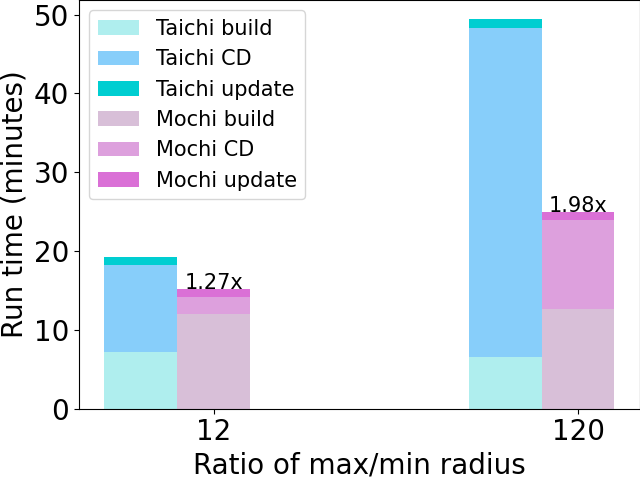}
            \caption{Runtime comparison}
            \label{fig:radii_plot}
        \end{subfigure} 
    \end{minipage}
    \caption{Performance comparison of Mochi against Taichi for rotating gravity DEM particle simulation of 8 seconds. $1$M particles with maximum to minimum radius ratio of $12$ and $120$ with radius in the range of $0.0005$-$0.006$ and $0.0005$-$0.06$, respectively. Then numbers on Mochi bars show Mochi speedups over baseline.}
    \Description{Performance comparison of Mochi against Taichi for rotating gravity DEM particle simulation of 8 seconds. $1$M particles with maximum to minimum radius ratio of $12$ and $120$ with radius in the range of $0.0005$-$0.006$ and $0.0005$-$0.06$, respectively. Then numbers on Mochi bars show Mochi speedups over baseline. The subsubsection referencing this image explains the trends.}
    \label{fig:1c_radii}
\end{figure}

A key observation is that Mochi’s DCD time increases at a slower rate than Taichi’s as the particle size increases.
In Taichi’s hashmap implementation, the size of each hash cell is set to four times the radius of the largest particle in the dataset --- just large enough to contain any collision pair.
As the radius of the largest particle increases, the number of hash cells decreases, which becomes problematic in scenarios with a wide range of particle sizes, such as the $0.0005$--$0.06$ setting.
In this case, several smaller particles are mapped to the same hash cells, increasing hashmap collisions and lookup cost.
In contrast, Mochi tailors the AABBs for individual particles, making its BVH more robust to radius non-uniformity.

Moreover, larger particles tend to collide with a greater number of neighbors.
In both Mochi and Taichi’s hashmap implementations, a single thread is assigned per particle; however, in the hashmap case, the thread corresponding to a large particle performs significantly more work.
This load imbalance does not arise in Mochi, since collision is detected by rays corresponding to the smaller particles, allowing work to be more evenly distributed.

\subsubsection{Bunny}
In this benchmark, we drop 14 bunny objects, where each bunny consists of $82{,}337$ spherical particles (Figure~\ref{fig:bunny_dataset}).
We simulate dropping $1$ and $14$ bunnies into unit and 2-unit cubes, respectively, for $20{,}000$ and $30{,}000$ iterations\footnote{Refer to the 32:56 mark of the video for a complete simulation.}.
Figure~\ref{fig:bunny_plot} shows the runtime comparison between Taichi and Mochi.
Mochi’s runtime is $11.06$x and $9.26$x lower than Taichi’s for the two bunny configurations.
Even in these highly clustered scenarios, Mochi’s RT-architecture-based BVH construction and collision detection outperform Taichi’s optimized shader kernel implementation.

\begin{figure}[h]
    \centering
    \begin{minipage}{\linewidth}
        \centering
        \begin{subfigure}{0.49\linewidth}
            \includegraphics[height=3.3cm,width=\linewidth]{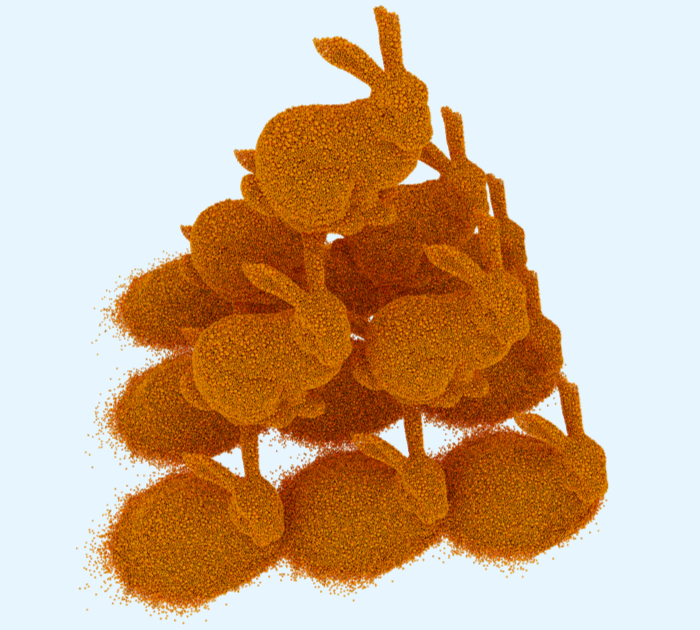}
            \caption{Frame 110 in simulation}
            \label{fig:bunny_dataset}
        \end{subfigure}
        \hfill
        \begin{subfigure}{0.49\linewidth}
            \includegraphics[width=\linewidth]{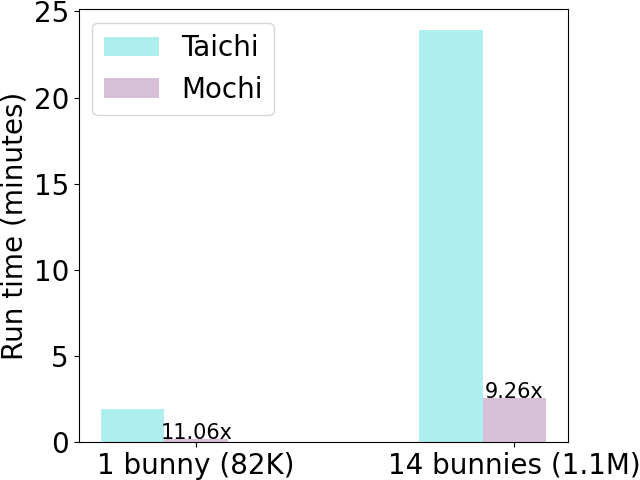}
            \caption{Runtime comparison}
            \label{fig:bunny_plot}
        \end{subfigure} 
    \end{minipage}
    \caption{Performance comparison of Mochi against Taichi for freefall DEM particle simulation of bunnies made of spherical particles. Then numbers on Mochi bars show Mochi speedups over baseline.}
    \Description{Performance comparison of Mochi against Taichi for freefall DEM particle simulation of bunnies made of spherical particles. Then numbers on Mochi bars show Mochi speedups over baseline. The subsubsection referencing this image explains the trends.}
    \label{fig:breakdown}
\end{figure}

\subsection{Performance Analysis}

\begin{table*} 
    \centering
    \caption%
      {Comparison of individual kernel and total times taken by Taichi and Mochi on millions of freely falling particles. All times in minutes.}
    \label{table:speedups}
    \begin{tabular}{l@{}rrrr rrrr rrrr} 
        \toprule

         \multirow{2}{*}{\raisebox{-\heavyrulewidth}{Number of }}
          & \multicolumn{8}{c}{Time in minutes} & \multicolumn{4}{c}{Speedups} \\
        \cmidrule(lr){2-9}

        particles  & \multicolumn{4}{c}{Taichi} & \multicolumn{4}{c}{Mochi}& \\
        \cmidrule(lr){2-5}\cmidrule(lr){6-9} \cmidrule(lr){10-13} 
        
        & \multicolumn{1}{c}{Build} 
        & \multicolumn{1}{c}{DCD}  
        & \multicolumn{1}{c}{Update} 
        & \multicolumn{1}{c}{Total} 
        & \multicolumn{1}{c}{Build} 
        & \multicolumn{1}{c}{DCD}  
        & \multicolumn{1}{c}{Update} 
        & \multicolumn{1}{c}{Total} 
        & \multicolumn{1}{c}{Build} 
        & \multicolumn{1}{c}{DCD}  
        & \multicolumn{1}{c}{Update}
        & \multicolumn{1}{c}{Total} 
        \\\midrule

1000000 & 3.3 & 29.49 & 0.58 & 33.37 & 1.63 & 5.44 & 0.49 & 7.56 & 2.02 & 5.42 & 1.19 & 4.41 \\
2000000 & 3.51 & 141.3 & 1.15 & 145.95 & 3.14 & 18.27 & 0.99 & 22.4 & 1.12 & 7.73 & 1.16 & 6.51 \\
3000000 & 3.7 & 294.57 & 1.7 & 299.97 & 4.19 & 44.5 & 1.5 & 50.18 & 0.88 & 6.62 & 1.13 & 5.98 \\
4000000 & 3.95 & 468.17 & 2.27 & 474.38 & 5.1 & 94.18 & 2.0 & 101.28 & 0.77 & 4.97 & 1.13 & 4.68 \\
5000000 & 4.2 & 638.28 & 2.83 & 645.31 & 5.78 & 190.08 & 2.5 & 198.36 & 0.73 & 3.36 & 1.13 & 3.25 \\
        \bottomrule
    \end{tabular}
    \label{table:results}
\end{table*}

In this subsection, we analyze Mochi’s performance gains and limitations by comparing it with SOTA DCD implementations across several parameters such as the number of particles, rebuild frequency, and GPU architecture.

\subsubsection{Where does the speedup come from?} 
\label{subsec:speedups}

Table~\ref{table:results} shows the breakdown of simulation times for Taichi and Mochi, along with the speedups achieved by Mochi over Taichi’s hashmap implementation for the freefall benchmark with $1$M particles described in Section~\ref{subsec:freefall}.
The time spent in all phases for both Taichi and Mochi (Columns $2$--$8$) increases as the number of particles increases.
Build time speedups (Column $9$) decrease, while DCD speedups (Column $10$) increase, and update-time speedups (Column $11$) remain nearly constant.
With the exception of the $1$M particles case, the total time speedups (Column $12$) decrease as the number of particles increases.
This is because Mochi’s build and DCD times are increasing at a higher rate than Taichi’s.
Specifically, Mochi’s total runtime increases by $2.96$x from $1$M to $2$M particles and by $1.96$x from $4$M to $5$M particles, whereas Taichi’s total runtime initially increases by $4.37$x and later by $1.36$x over the same ranges.
We attribute this decrease in the rate of Taichi's simulation time and the corresponding increase in Mochi’s speedup at $2$M particles to Taichi’s offline caching, which reuses previously compiled kernels to optimize subsequent executions.

Although the decreasing speedup trend raises the question of whether Mochi would continue to outperform Taichi at much larger scales, Taichi is more likely to face out-of-memory errors before outperforming Mochi, due to the rapid increase in the number of collision pairs. 
However, it is challenging to empirically quantify beyond $5$M particles in practice, as simulations at this scale run for several hours.

\paragraph{Build time speedups} 
Build times increase for both Mochi and Taichi as the number of particles increases. 
However, Mochi’s build-time speedups diminish and disappear beyond $3$M particles, owing to poorly refit BVH.
In this experiment, Mochi refits the BVH at every iteration.
During freefall, particles repeatedly collide with the ground and bounce upon touching the ground.
Higher the number of particles, more they bounce, larger the relative displacement and the number of collisions.
These effects degrade the quality of the refitted BVH, reducing the benefit of refitting at larger scales.
Refitting exposes the inherent trade-off between BVH maintenance cost and traversal quality, a limitation shared by all BVH-based prior work.


\subsubsection{Performance across scale}
Figure~\ref{fig:10power} plots the runtimes of Taichi and Mochi for the freefall simulation as the number of particles increases from $1$K to $1$M in factors of $10$.
Taichi’s runtime barely changes until it reaches $1$M particles, indicating that its execution time is dominated by overhead costs for particle counts below $100$K.
In contrast, Mochi’s runtime increases steadily with an increase in the number of particles.
These differing trends result in a dip in speedup around $100$K particles.


\begin{figure}[h]
    \centering
    \includegraphics[width=0.5\linewidth]{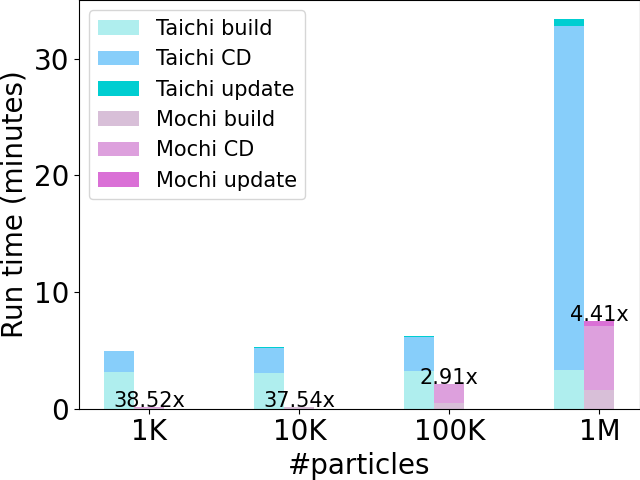}
    \caption{Performance comparison of Mochi against Taichi for freefall DEM particle simulation as the number of particles increases at a factor of 10. Then numbers on Mochi bars show Mochi speedups over baseline.}
    \Description{Performance comparison of Mochi against Taichi for freefall DEM particle simulation as the number of particles increases at a factor of 10. Then numbers on Mochi bars show Mochi speedups over baseline. The subsubsection referencing this image explains the trends.}
    \label{fig:10power}
\end{figure}

\subsubsection{Exploring Mochi's index times}
\label{subsec:build_heur}
Figure~\ref{fig:build_heur} shows how Mochi’s build and DCD times change as the number of iterations between BVH rebuilds varies during the rotating gravity simulation with $1$M particles.
After the initial build of BVH, it is rebuilt every $n$-th iteration and refitted in the remaining iterations.
As the rebuild frequency decreases, both the total build time and the overall runtime decrease, since BVH refitting is significantly cheaper than rebuilding.
However, when rebuilds are too infrequent, repeated refitting degrades BVH quality, leading to increased DCD time and, consequently, higher total runtime.

\begin{figure}[h]
    \centering
    \includegraphics[width=0.5\linewidth]{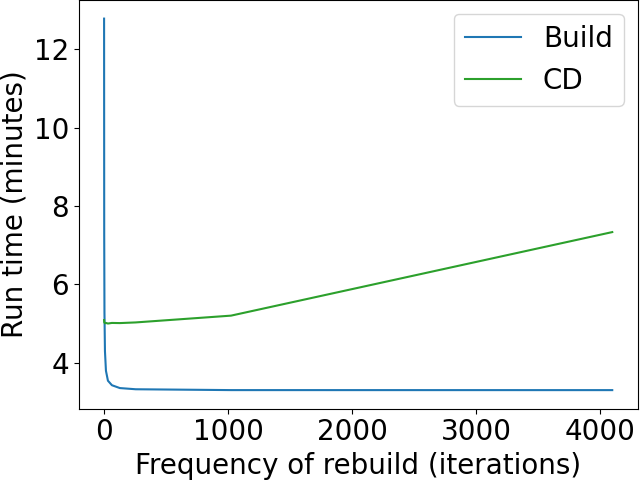}
    \caption{Build and DCD times of Mochi as the number of iterations between two rebuilds vary.}
    \Description{Build and DCD times of Mochi as the number of iterations between two rebuilds vary. The subsubsection referencing this image explains the trends.}
    \label{fig:build_heur}
\end{figure}

\subsubsection{Performance on a different RT hardware}

\begin{table}[h!]
    \centering
    \caption%
    {Hardware specifications of RTX 4070Ti and RTX 6000 GPUs}%
    \label{table:io_triton}
    \begin{tabular}{l@{}rrr} 
        \toprule
        & RTX 4070ti & RTX 6000 & Increase factor\\\midrule
        RT cores & 80 & 142 & 1.78x\\
        Shader cores & 7,680 & 18,176 & 2.37x\\
        Memory & 12GB & 48GB & 4x\\
        \bottomrule
    \end{tabular}
\end{table}

To demonstrate that Mochi’s implementation generalizes across different RTX GPUs, we conducted a freefall simulation with $1$M particles on an RTX 6000.
Table~\ref{table:io_triton} summarizes the hardware differences between the GPUs.
Mochi achieves speedups over Taichi of $1.51$x, $5.55$x, $1.44$x, and $3.96$x for build, DCD, update, and total runtimes, respectively.
Compared to the RTX 4070Ti, the RTX 6000 provides a larger increase in shader cores ($2.37$x) than in RT cores ($1.78$x), yet Mochi achieves comparable DCD speedups.
In contrast, build-time and overall speedups are slightly lower than those observed on the RTX 4070Ti (Table~\ref{table:speedups}), indicating that BVH construction benefits less from the increased number of shader cores.

\section{Discussion}
\label{sec:discussion}

\paragraph{Numerical robustness.}
To assess the robustness of Mochi under floating point inconsistencies, we performed a sensitivity analysis by perturbing particle positions with small random noise. 
Specifically, we generated 30 perturbed variants of the $1$M particles dataset (described in Section~\ref{sec:evaluation} by introducing noise sampled uniformly in the range $[1e{-6}, 1e{-5}]$ to particle positions.
After performing freefall simulation on the perturbed datasets, we compared the per-timestep distribution of detected collisions produced by Mochi and Taichi. 
While the exact number of collision pairs may differ due to floating-point operations, the collision count distributions were indistinguishable between the two implementations.
This indicates that Mochi exhibits numerically stable behavior and performs the same logical collision detection task as existing GPU-based approaches, despite differences in execution model.

\paragraph{Geometries beyond spheres.}
In principle, the narrow phase of Mochi (implemented in the \texttt{OPTIX\_INTERSECT} kernel) could be programmed to handle other geometric primitives, provided that robust object--object intersection tests are available.
However, extending the collision detection beyond spheres is primarily a mathematical challenge rather than a computational one.
Even for simple generalizations such as ellipsoids, detecting collisions when  ellipsoids are oriented along arbitrary axes requires solving complicated quartic equations, which is computationally expensive and numerically unstable~\cite{ellipsoids}. 
As a result, approximating complex geometries using spheres remains a widely used optimization. 
By leveraging such approximations, Mochi can be utilized to support a broad range of non-spherical objects without modifying its core execution model.

\paragraph{Limitations.}
While Mochi’s design can be extended beyond spheres and to other hardware platforms, the strongest claim supported by our implementation and evaluation is that Mochi correctly performs DCD for rigid spheres on NVIDIA RT hardware.
Mochi detects collisions at discrete time steps and may therefore miss collisions that would be detected by continuous collision detection algorithms.
Mochi treats each particle as a rigid spherical object and does not account for deformations. 
Mochi operates using single-precision floating point arithmetic on GPU RT architecture; as with other large-scale GPU simulations, numerical discrepancies may arise, although the algorithm itself is deterministic and theoretically sound.

\paragraph*{Limitations imposed by underlying RT architecture.}
The BVH constructed by RT architecture is opaque to the programmer, making it difficult to directly assess BVH quality or determine when refitting is preferable to rebuilding.
In Section~\ref{subsec:build_heur}, we analyzed the  BVH maintenance cost by varying rebuild frequency; however, standard quality metrics such as surface area heuristic (SAH) are not accessible.
Additionally, current RT hardware does not support double-precision (FP64) operations, so Mochi operates in single precision (FP32) and considers only FP32 input.
If FP64 precision is required, a possible approach is to slightly expand the AABBs of proxy spheres so that no valid collisions are missed. 
This may introduce a few additional false positives in the broad phase, but these will be filtered out in the narrow phase, which runs on shader cores that use FP64 arithmetic for collision checks. 
Although our implementation uses OptiX, Mochi’s reduction relies only on efficiently reusing BVH traversal and ray–AABB intersection test functionalities provided by the underlying RT architecture, and is not tied to OptiX interface semantics.



\section{Related Work}
\label{sec:relatedwork}
In this section, we give an overview of Collision Detection (CD) works using GPUs, BVH, and ray tracing. We refer to the surveys~\cite{cdsurvey,bvhsurvey} for an in-depth review.

\paragraph*{Collision detection on GPUs}
Due to the computational intensity of CD, numerous algorithms and optimization methods have been developed to harness the parallel computing capabilities of GPUs. 
Spatial subdivision data structures such as uniform and hierarchical grids are faster to build since they are regular, compared to BVH-based approaches. 
Wong \etal~\cite{octree_2014} used an octree represented as a grid on GPU, while Weller \etal~\cite{kdet_2017} used hierarchical grids and hybrid spatial hashing. 
BVH-based approaches on GPU employ several optimizations to mitigate the performance loss due to the inherent irregular nature of BVH traversal.
A common approach is to use a fast BVH construction algorithm and then use auxiliary data structures during the traversal to make up for the poor BVH quality.
For example, Lauterbach \etal~\cite{lauterbach2010gproximity} used  tight-fitting BVHs and parallel tree-front-based BVH traversal, whereas Wang \etal~\cite{wang2018} used stackless depth-first search-based traversal, and tree front logs.
Chitalu \etal~\cite{oibvh} optimized the accelerating structures by building a specialized BVH to speed up BVH node access during traversal.

\paragraph{Neighbor Search using RT architecture}
Multiple researchers have attempted to accelerate proximity queries like k-nearest neighbor search and radius neighbor search using RT architectures. 
Zellman \etal~\cite{zellman} proposed a method where radius neighbor searches compute forces between point masses in graph drawing. 
Evangelou \etal~\cite{evangelou} applied a similar method to perform photon mapping. 
Zhu \etal~\cite{rtnn} and Nagarajan \etal~\cite{trueknn} further optimized this method for nearest neighbor applications using query re-ordering and bundling, and iteratively increasing the search radius, respectively.
However, neither of their optimizations are transferable to DCD as their optimizations focus on finding $k$ nearest neighbors rather than all radius neighbors needed in DCD.

\paragraph{Collision detection using RT hardware}
With the recent introduction of RT cores in GPUs, Vassilev \etal~\cite{vassilev_2021} and Zhao \etal~\cite{Zhao_arbitrary} proposed using RT architecture to perform velocity-based cloth simulations and large-scale DEM simulations of arbitrarily shaped particles, respectively.
Zhao \etal~\cite{zhao2023leveraging}, which is closest to our work, uses NNS reduction~\cite{evangelou} to perform DCD during DEM particle simulation.
Sui \etal~\cite{roboticspaper} use sphere-based robot approximations and swept trajectories to perform discrete-pose and continuous collision detection.

\paragraph{Collision detection using RT algorithm}
Most of the RT-based CD approaches that did not use RT hardware focus on accelerating the narrow phase.
For a given pair of potentially colliding volumetric deformable bodies, Hermann \etal~\cite{hermann} detect the collisions by shooting rays from surface vertices in the direction of their inward normals.
Kim \etal~\cite{kim_mesh} support deep penetrations and non-convex objects.
To reduce the number of rays cast for rigid convex and concave objects, Lehericey \etal introduce iterations over Hermann \etal algorithm, launch predictive rays and introduce a pipeline that selects an algorithm based on the object characteristics~\cite{lehericey_vrst_2013,lehericey_evge_2013, lehericey_grapp_2015}.
Even though their algorithm improves CD speed, it is dependent on the collisions detected initially and the approximations it employs.

\paragraph{BVH construction}
Our work relies on RT cores for both BVH construction and traversal, whereas much prior work focuses on optimizing these algorithms. 
Optimizing BVH construction is particularly important in highly dynamic, real-time scenes, where simple refitting can degrade performance due to reduced BVH quality. 
For instance, Linear BVH (LBVH) aims to minimize the construction cost while maintaining good BVH quality~\cite{Apetrei14}, and its combination with SAH balances construction and traversal efficiency~\cite{Lauterbach09}. 
Implicit approaches further reduce memory usage (and, in turn, build time) at the expense of more complicated tree access~\cite{oibvh}.

\paragraph{Continuous collision detection and others}
Continuous Collision Detection (CCD) algorithms find collisions between discrete time steps that DCD algorithms may not be able to find.
Brochu \etal~\cite{brochu2012efficient} proposed a geometrically exact detection for adaptive cloth simulation and used conservative culling and interval filters to improve performance.
Wang \etal~\cite{selfcollision} presented a CCD as well as a DCD algorithm for detecting self-collisions in deformable models using a normal cone test.
We do not currently explore how Mochi can be paired with CCD or self-collision algorithms.

\section{Conclusion}
\label{sec:conclusion}
Mochi shows that rethinking how collision detection is reduced to RT primitives, rather than optimizing existing pipelines, is essential for effectively exploiting RT architecture for non-rendering workloads.
We introduce a new RT-based formulation for accelerating discrete collision detection, despite the lack of access to the data structures built by the architecture.
Across all of our benchmarks, Mochi achieves consistent speedups, even for simulations involving up to $5$ million particles. 
At this scale, simulations often run for hours to produce only a few seconds of output, making even modest performance improvements practically significant.
Graphics libraries such as Taichi, which target high-performance simulation workloads, can benefit from incorporating Mochi’s RT-based DCD formulation into their GPU kernel implementations.
As RT hardware continues to be adopted across GPU vendors, and our approach relies only on core functionalities such as BVH construction and ray traversal, Mochi is portable across RT APIs and is likely to benefit from continued improvements in each new RT generation.


\begin{acks}
We thank the anonymous reviewers for their valuable feedback. We are also grateful to Sunidhi Bachawala for valuable discussions on the dynamics of particle systems and motivating use cases, and to Raghav Malik for verifying the proofs. This work was funded by NSF grants CCF-1908504, CCF-1919197 and CCF-2216978.
\end{acks}

\bibliography{references}

\end{document}